\documentclass[12pt]{article}
\usepackage[utf8]{inputenc}
\usepackage[english]{babel}

\usepackage{amssymb}
\usepackage{amsmath}
\usepackage{amsfonts}
\usepackage{game}
\usepackage{subfig}
\newcommand{\marg}{\text{marg}}
\usepackage{natbib}
\usepackage{accents}
\usepackage{hyperref}
\usepackage{tikz}
\usetikzlibrary{calc,positioning}
\usepackage{tikz-cd}
\usetikzlibrary{patterns}   
\usetikzlibrary{intersections}   
\usetikzlibrary{automata}   
\usetikzlibrary{patterns} 
\usetikzlibrary{arrows,chains,matrix,positioning,scopes}  
\usetikzlibrary{shapes,backgrounds}
 
\newtheorem{theorem}{Theorem}

\newtheorem{lemma}{Lemma}

\newtheorem{exmp}{Example}
\newtheorem{corollary}{Corollary}
\newtheorem{claim}{Claim}
\newenvironment{proof}[1][Proof]{\noindent\textbf{#1} }{\ \rule{0.5em}{0.5em}}
\usepackage[skins]{tcolorbox}
\tcolorboxenvironment{theorem}{blanker, before skip=10pt, after skip=10pt}

\usepackage[margin=1in]{geometry}
\usepackage[all]{xy}
\usepackage{multicol}

\begin{document}

	\title{Payoff Continuity in Games of Incomplete Information Across Models of Knowledge\begin{footnote}{An associate editor and three anonymous referees provided valuable feedback that improved the quality of the manuscript. Aislinn Bohren, Eduardo Faingold, Ben Golub, Tetsuya Hoshino, Scott Kaplan, Annie Liang, George Mailath, Stephen Morris, Ricardo Serrano-Padial, Juuso Toikka, and Yuichi Yamamoto provided helpful comments, suggestions, and guidance. \href{https://www.refine.ink/}{Refine.ink} was used to proofread the paper for consistency and clarity. The views expressed here are solely my own and do not in any way represent the views of the U.S. Naval Academy, U.S. Navy, or the Department of Defense.}\end{footnote}}

	\author{Ashwin Kambhampati\begin{footnote}{Department of Economics, United States Naval Academy. Email: \texttt{kambhamp@usna.edu}. Postal Address: Michelson 361, 572M Holloway Road,
Annapolis, MD 21402.}\end{footnote}} 
	\maketitle
	\thispagestyle{empty}
 \vspace{-10mm}
	\begin{abstract}
    Equilibrium predictions in games of incomplete information are sensitive to the assumed information structure. \cite{MS1996} and \cite{kajii1998payoff} define topological notions of proximity for common prior information structures such that two information structures are close if and only if  (approximate) equilibrium payoffs are close.  However, \cite{MS1996} fix a common prior and define their topology on profiles of partitions over a state space, whereas \cite{kajii1998payoff} define their topology on common priors over the product of a state space and a type space. We prove the open conjecture that two partition profiles are close in the \cite{MS1996} topology if and only if there exists a labeling of types such that the associated common priors are close in the \cite{kajii1998payoff} topology.
	\end{abstract}

    

\newpage
\pagenumbering{arabic}
\section{Introduction}

A game of incomplete information consists of a set of players, an action set and payoff function for each player, and an information structure. How does the set of Bayesian Nash equilibrium payoffs change as the information structure changes? \cite{rubinstein1989electronic}'s Email Game illustrates a striking payoff discontinuity; there can exist an equilibrium under a common knowledge information structure yielding expected payoffs that are not approximated in \textit{any} equilibrium under an information structure in which there are arbitrarily many, but finite, levels of mutual knowledge.

 \cite{MS1996} and \cite{kajii1998payoff} identify coarse topologies on common prior information structures that preserve continuity of $\epsilon$- Bayesian Nash equilibrium payoffs across all bounded games of incomplete information.\begin{footnote}{This exercise is uninformative if one requires that players exactly optimize ($\epsilon=0$). Theorem 5 of \cite{peskietal} considers a general space of information structures and shows that the topology that preserves payoff continuity of (exact) Bayesian Nash equilibria on this space is the discrete topology.}\end{footnote} Both topologies demonstrate that common $p$-belief (\cite{monderer1989approximating}) is the appropriate relaxation of common knowledge to preserve continuity of equilibrium payoffs. However, \cite{MS1996} and \cite{kajii1998payoff} model the proximity of information structures differently. \cite{MS1996} fix a state space and a common prior over it, and consider the differences in beliefs induced by a change in partitions over the state space. \cite{kajii1998payoff} fix state and type spaces, and consider the differences in beliefs induced by a change in the common prior over its product. 

The relationship between the two modeling approaches and the resulting topologies has remained an open question. As \cite{kajii1998payoff} write,
\begin{quote}
	Our characterization of the proximity of information has a similar flavor to Monderer and Samet's, but we have not been able to establish a direct comparison. By considering a fixed type space, we exogenously determine which types in the information systems correspond to each other. In the Monderer and Samet approach, it is necessary to work out how to identify types in the two information systems. Thus we conjecture that two information systems are close in Monderer and Samet's sense if and only if the types in their construction can be labelled in such a way that the information systems are close in our sense.
\end{quote}
In this paper, we prove \cite{kajii1998payoff}'s conjecture by constructing distance-preserving maps from partition profiles to common priors.

\section{Setting}\label{prelim}

Fix a set of two or more players $\mathcal{N}:= \{1,2, \ldots ,N\}$ and a probability triple $(S, \Sigma, P)$, where $S$ is an infinite state space, $\Sigma$ is a $\sigma$-algebra of events, and $P$ is a probability measure. The state space $S$ can be taken to be countable or uncountable. In what follows, all countable sets are equipped with the discrete $\sigma$-algebra and all product sets are equipped with the product $\sigma$-algebra. 

We introduce here some useful notation. Let $X$, $Y$, and $Z$ be measure spaces. For any measurable function $g: X \rightarrow Y$, let $g(E)$ denote the image of $E \subseteq X$ under $g$ and $g^{-1}(F)$ denote the pre-image of $F \subseteq Y$ under $g$. If $h: X \rightarrow Z$ is another measurable function, then let $g \times h: X \rightarrow Y \times Z$ be the measurable function defined by $(g \times h)(x)= (g(x), h(x))$ for all $x \in X$. Finally, let $\iota: S \rightarrow S$ be the identity function on the set of states, i.e., $\iota(s)=s$ for all $s \in S$. 
 
\subsection{The MS Topology}

Denote by $\mathcal{P}$ the set of partitions of $S$ into non-null elements of $\Sigma$. Note that every element of $\mathcal{P}$ is thus a countable partition. Denote by $\mathcal{P}^N$ the set of all partition profiles, with typical element denoted by $\Pi=(\Pi_1, \ldots ,\Pi_N)$.
We define the \cite{MS1996} topology on $\mathcal{P}^N$.

Denote by $\Pi_i(s)$ player $i$'s partition element containing a state $s \in S$. At a state $s \in S$, player $i$ assigns probability $P(E|\Pi_i(s))$ to the event $E \subseteq S$. Player $i$ \textbf{$p$-believes} an event $E \subseteq S$ at a state $s \in S$ if $P(E|\Pi_i(s)) \geq p.$ Denote by $B^p_{\Pi_i}(E)$ the set of states at which player $i$ $p$-believes $E$ under the partition $\Pi_i$. The set of states at which $E$ is \textbf{mutual $p$-belief} is $B^p_\Pi(E):= \underset{i \in \mathcal{N}}{\cap} B^p_{\Pi_i}(E)$. The set of states at which $E$ is \textbf{$m$-level mutual $p$-belief} is $(B^p_\Pi)^m(E)$, the $m$-th iteration of $B^p_\Pi(\cdot)$ over the set $E$. Finally, the set of states at which $E$ is \textbf{common $p$-belief} is $C^p_\Pi(E) := \underset{m \geq 1}{\cap} (B^p_\Pi)^m(E)$.

Define $I_{\Pi, \Pi'}(\epsilon)$ to be the set of states at which the conditional symmetric difference between each player's partition elements containing that state is less than $\epsilon$,
$$I_{\Pi, \Pi'}(\epsilon):= \underset{i \in \mathcal{N}}{\cap} \{s \in S: \max\{P(\Pi_i(s) \backslash \Pi'_i(s) | \Pi_i(s)), P(\Pi'_i(s) \backslash \Pi_i (s) | \Pi'_i(s))\} \leq \epsilon\}.$$
Define
$$d^{MS}(\Pi,\Pi'):=\max\{d^{MS}_1(\Pi,\Pi'), d^{MS}_1(\Pi',\Pi)\},$$
where $d^{MS}_1(\Pi,\Pi')$ is an ex-ante measure of states in $I_{\Pi,\Pi'}(\epsilon)$ at which the event $I_{\Pi, \Pi'}(\epsilon)$ is common $(1-\epsilon)$-belief:
$$d^{MS}_1(\Pi,\Pi'):= \inf \{ \epsilon : P(C^{1-\epsilon}_{\Pi'}(I_{\Pi, \Pi'}(\epsilon)) \cap I_{\Pi,\Pi'}(\epsilon)) \geq 1- \epsilon\}.$$

Using $d^{MS}$, we now define the \textbf{MS topology}. For each $\Pi \in \mathcal{P}^N$ and $r>0$, define the $r$-neighborhood of $\Pi$ by \[B_{d^{MS}}(\Pi, r):= \{ \Pi' \in \mathcal{P}^N: d^{MS}(\Pi, \Pi') < r \}. \] Then, a subset of partition profiles, $O \subseteq \mathcal{P}^N$, is open in the MS topology if and only if, for all $\Pi \in O$, there exists an $r>0$ such that $B_{d^{MS}}(\Pi, r) \subseteq O$. 

We remark here that $d^{MS}$ is not the distance defined in \cite{MS1996}. \cite{MS1996} replace the common $p$-belief operator with the stronger ``joint common repeated $p$-belief" operator and bound $I_{\Pi,\Pi'}(\epsilon)$ by $1/2$ so that their distance satisfies the triangle inequality (and is therefore a pseudo metric). Nevertheless, $d^{MS}$ induces the same topology as the pseudo metric defined in \cite{MS1996} by Theorem 5.2 of their paper.

\subsection{The KM Topology}
Fix a type space $T :=T_1 \times \cdots \times T_N$, where $T_i$ is a countably infinite set of types for player $i \in \mathcal{N}$. We call a pair $(s,t) \in S \times T$ a \textbf{KM state}. Denote by $\Delta(S \times T)$ the set of probability measures over KM states, i.e., the set of common priors. \cite{kajii1998payoff} define their topology on the space of all common priors, assuming $S$ is a countably infinite set. We define the trace of the \cite{kajii1998payoff} topology, extended to allow for a potentially uncountable state space $S$, on the subset of partitional information structures, i.e., those generated by the fixed prior over states $P \in \Delta(\Theta)$ and some partition profile $\Pi \in \mathcal{P}^N$.   

We begin by defining the \cite{kajii1998payoff} topology on the space of all common priors, allowing for a potentially uncountable state space $S$. Denote by $\mu(t_i)$ the marginal distribution of $\mu$ on $T_i$ evaluated at $t_i$. If $\mu(t_i)>0$, then the probability of an event $E \subseteq S \times T$ conditional on player $i$'s type $t_i$ is denoted by $\mu(E |t_i):= \mu(E \cap (S \times \{t_i\} \times T_{-i}))/ \mu(t_i)$. Player $i$ \textbf{$p$-believes} an event $E \subseteq S \times T$ at $(s,t) \in S \times T$ if $\mu(E|t_i) \geq p$ or $\mu(t_i)=0$. Denote by $B^p_{\mu_i}(E)$ the set of all KM states, $(s,t)$, at which player $i$ $p$-believes event $E$. The set of KM states at which $E$ is \textbf{mutual $p$-belief} is $B^p_\mu(E):= \underset{i \in \mathcal{N}}{\cap} B^p_{\mu_i}(E) .$ The set of KM states at which $E$ is \textbf{$m$-level mutual $p$-belief} is $(B^p_\mu)^m(E)$, the $m$-th iteration of $B^p_\mu(\cdot)$ over the set $E$. The set of KM states at which $E$ is \textbf{common $p$-belief} is $C^p_\mu(E):= \underset{m \geq 1}{\cap} (B^p_\mu)^m(E)$. 

We now define the set of KM states at which each player has conditional beliefs that differ by at most $\epsilon$ over any event given their type. Specifically, let $A_{\mu, \mu'}(\epsilon)$ be the set of KM states $(s,t) \in S \times T$ such that
\begin{enumerate}
	\item $\mu(t_i)>0$ and $\mu'(t_i)>0$, and
	\item $|\mu(E \times F| t_i)- \mu'(E \times F | t_i)| \leq \epsilon$ for all events $E \times F \subseteq S \times T$.\begin{footnote}{Notice that the sup-norm difference between $\mu$ and $\mu'$ over all rectangles $E \times F \subseteq S \times T$ is equivalent to the sup-norm difference between $\mu$ and $\mu'$ over all events because rectangles generate the product $\sigma$-algebra on $S \times T$.}\end{footnote}
\end{enumerate}
Now, define a function mapping pairs of common priors to real numbers: $$\rho^{KM}(\mu,\mu'):= \max\{\rho^{KM}_1(\mu,\mu'), \rho^{KM}_1(\mu',\mu), \rho^{KM}_0(\mu,\mu')\},$$ where $\rho^{KM}_1(\mu,\mu')$ is an ex-ante measure of KM states under which the event $A_{\mu, \mu'}(\epsilon)$ is common $(1-\epsilon)$-belief,
$$\rho^{KM}_1(\mu,\mu') := \inf \{ \epsilon: \mu' (C^{1-\epsilon}_{\mu'} (A_{\mu, \mu'}(\epsilon))\cap A_{\mu, \mu'}(\epsilon)) \geq 1-\epsilon\},$$
and $\rho^{KM}_0(\mu,\mu')$ is the sup-norm distance between $\mu$ and $\mu'$,
$$\rho^{KM}_0(\mu,\mu'):= \underset{E \times F \subseteq S \times T}{\sup}~|\mu(E \times F)-\mu'(E \times F)|.$$ 
For each $\mu \in \Delta(S \times T)$ and $r>0$, define the $r$-neighborhood of $\mu$ by \[B_{\rho^{KM}}(\mu, r):= \{ \mu' \in \Delta(S \times T): \rho^{KM}(\mu, \mu') < r \}. \] Then, a subset of common priors, $O \subseteq \Delta(S \times T)$, is open in the extended \cite{kajii1998payoff} topology if and only if, for all $\mu \in O$, there exists an $r>0$ such that $B_{\rho^{KM}}(\mu, r) \subseteq O$.\begin{footnote}{In the definition of $\rho^{KM}_1(\mu,\mu')$, we require
\[\mu' (C^{1-\epsilon}_{\mu'} (A_{\mu, \mu'}(\epsilon)) \cap A_{\mu, \mu'}(\epsilon)) \geq 1-\epsilon,\]
whereas \cite{kajii1998payoff} require
\[\mu' (C^{1-\epsilon}_{\mu'} (A_{\mu, \mu'}(\epsilon))) \geq 1-\epsilon.\] 
Nevertheless, Equation (5.11) in \cite{MS1996} establishes
		\[\mu' (C^{1-\epsilon}_{\mu'} (E) \cap E) \geq (1-2 \epsilon) \mu' (C^{1-\epsilon}_{\mu'} (E))\]
		for any $E \subseteq S \times T$. It follows that the two distances induce the same topology.}\end{footnote} 

Using $\rho^{KM}$, we define the trace of the KM topology on the subspace of partitional information structures. 
More precisely, given a measurable function $g: S \rightarrow T$, say that a common prior $\mu \in \Delta(S \times T)$ is \textbf{$g$-consistent} if, for any event $E \subseteq S \times T$,
\begin{equation} \mu(E) = P((\iota \times g)^{-1}(E)). \label{pushforward} \end{equation}
Say that a measurable function $\tau: S \rightarrow T$ is a \textbf{$\Pi$-labeling} if, for any $s,s' \in S$, $\tau_i(s)=\tau_i(s')$ if and only if $\Pi_i(s)=\Pi_i(s')$. Then, a common prior $\mu$ is \textbf{$\Pi$-consistent} if it is $\tau$-consistent under the $\Pi$-labeling $\tau: S \rightarrow T$. Let $\mathcal{C}$ denote the set of common priors $\mu \in \Delta(S \times T)$ for which there exists a $\Pi \in \mathcal{P}^N$ such that $\mu$ is $\Pi$-consistent, i.e., the space of partitional information structures.

We now define the trace of the extended \cite{kajii1998payoff} topology on $\mathcal{C}$. Specifically, let $d^{KM}:=\rho^{KM}|_{\mathcal{C} \times \mathcal{C}}$ be the restriction of $\rho^{KM}$ to the domain $\mathcal{C} \times \mathcal{C}$, i.e., for any $\mu, \mu' \in \mathcal{C}$, $d^{KM}(\mu, \mu')= \rho^{KM}(\mu, \mu')$. Moreover, for each $\mu \in \mathcal{C}$ and $r>0$, define the $r$-neighborhood of $\mu$ in $\mathcal{C}$ under $d^{KM}$ by \[B_{d^{KM}}(\mu, r):= \{ \mu' \in \mathcal{C}: d^{KM}(\mu, \mu') < r \}. \] Then, a subset of common priors, $O \subseteq \mathcal{C}$, is open in the trace of the extended \cite{kajii1998payoff} topology on $\mathcal{C}$ (henceforth, the \textbf{KM topology}) if and only if, for all $\mu \in O$, there exists an $r>0$ such that $B_{d^{KM}}(\mu, r)\subseteq O$.

\subsection{Payoff Continuity Results}\label{sec_payoffcont_results}

We point out here that \cite{MS1996} and \cite{kajii1998payoff} prove payoff continuity in their respective topologies with respect to different classes of games and with respect to different notions of proximity of $\epsilon$- Bayesian Nash equilibrium payoffs.

Regarding the classes of games, let $A_i$ be an action set for player $i$ that is finite and has at least two elements. In addition, let $A:= A_1 \times \cdots \times A_N$ be the set of all action profiles. Theorem 5.4 of \cite{MS1996} establishes upper- and lower- hemicontinuity of $\epsilon$- Bayesian Nash equilibrium payoffs in their topology across all bounded games in which each player $i$ has a utility function $u_i: A \times S \rightarrow \mathbb{R}$. In contrast, Proposition 8 of \cite{kajii1998payoff}, whose proof is corrected by \cite{rothschild2005payoff},\begin{footnote}{Specifically, \cite{rothschild2005payoff} shows that Lemma 4 can be replaced with a stronger Lemma that makes the proof of Proposition 5 correct as stated.}\end{footnote} establishes upper- and lower- hemicontinuity of $\epsilon$- Bayesian Nash equilibrium payoffs across all bounded games in which each player $i$ has a utility function $u_i: A \times S \times T \rightarrow \mathbb{R}$ and $S$ is assumed to be countable.
Although $S$ is taken to be countable in \cite{kajii1998payoff} and \cite{rothschild2005payoff}, this restriction plays no essential role. The corrected proof of Proposition 8 extends verbatim to the case in which $(S, \Sigma)$ is an arbitrary measurable space and all (bounded) utility functions are measurable with respect to $\Sigma$, once sums are replaced with Lebesgue integrals.\begin{footnote}{Note that the class of games for which Proposition 8 holds thus depends on $\Sigma$, the $\sigma$-algebra on $S$.}\end{footnote} In particular, the extended \cite{kajii1998payoff} topology introduced in this paper is the coarsest topology preserving $\epsilon$- Bayesian Nash equilibrium payoff continuity across all bounded games with a possibly non-countable state space and utility functions that depend on both states and types.

An additional difference between the two results concerns the measurement of the distance between payoffs. \cite{MS1996} measure the distance between payoffs by taking the expectation of absolute differences, whereas \cite{kajii1998payoff} measure it by the absolute difference of expectations. Because the absolute difference of expectations is smaller than the expectation of absolute differences, i.e., for any bounded random variables $X$ and $Y$,
\[| \mathbb{E}[X]-\mathbb{E}[Y] | \leq \mathbb{E}[|X-Y|], \] the continuity requirement in \cite{MS1996} is stronger than in \cite{kajii1998payoff}.\begin{footnote}{To prove the inequality, simply note that, by linearity of expectations, $| \mathbb{E}[X]-\mathbb{E}[Y] | =| \mathbb{E}[X-Y] |$. Then, because the absolute value function is convex, Jensen's inequality implies that $| \mathbb{E}[X-Y] | \leq \mathbb{E}[|X-Y|]$.}\end{footnote}

Even though \cite{kajii1998payoff} allow for a larger class of games in which types enter each player's utility function, our main result shows that convergence in the MS topology implies convergence in the extended KM topology whenever type profiles with similar beliefs across partition profiles are mapped to the same element of $T$. The crucial property of the mapping is that when a player's ``belief type" is perturbed in the MS topology, her ``payoff type" remains unchanged. Thus, differences in payoffs cannot be artificially generated through the utility function in the \cite{kajii1998payoff} framework. Conversely, despite the stronger continuity requirement in \cite{MS1996}, our main result also shows that convergence of common priors in the KM topology implies convergence of any ``belief consistent" sequence of partition profiles in the MS topology.

\section{Main Result}

The main result, Theorem \ref{thm_1}, establishes the existence of a distance-preserving map with respect to every partition profile, $\Pi$, in the MS topology. Specifically, the map sends every partition profile $\Pi'$ to a $\Pi'$-consistent common prior $\mu'$. In addition, partition profile labelings are defined so that each partition element in $\Pi'_i$ that is ``close" to a partition element in $\Pi_i$ is mapped to the same type, $t_i$. This property, which we call the common support condition, ensures that partition profiles $\Pi$ and $\Pi'$ are close in the MS topology if and only if the associated common priors $\mu$ and $\mu'$ are close in the KM topology.

\begin{theorem}\label{thm_1}
Fix a partition profile $\Pi \in \mathcal{P}^N$ and a number $\gamma \in (0,1/2)$. The following properties hold:
\begin{enumerate}
\item There exists a function $f_{(\Pi,\gamma)}: \mathcal{P}^N \rightarrow \Delta(S \times T)$ such that, for every $\Pi' \in \mathcal{P}^N$, $f_{(\Pi,\gamma)}(\Pi')$ is $\Pi'$-consistent, and, for any $\epsilon \in (0,\gamma]$,
\begin{enumerate}
    \item if $d^{MS}(\Pi, \Pi') \leq \epsilon$, then $d^{KM}(f_{(\Pi,\gamma)}(\Pi),f_{(\Pi,\gamma)}(\Pi')) \leq \epsilon$, and
    \item if $d^{KM}(f_{(\Pi,\gamma)}(\Pi),f_{(\Pi,\gamma)}(\Pi')) \leq \epsilon$, then $d^{MS}(\Pi, \Pi') \leq (N+1) \epsilon$.
\end{enumerate}

\item For any $\epsilon \in (0,\gamma]$, if there exists a $\Pi$-consistent common prior $\mu$ and a $\Pi'$-consistent common prior $\mu'$ such that $d^{KM}(\mu ,\mu') \leq   \epsilon$, then $d^{MS}(\Pi, \Pi') \leq (N+1) \epsilon$.
\end{enumerate}
\end{theorem}

The first part of Theorem \ref{thm_1} establishes that both $f_{(\Pi,\gamma)}$ and its inverse $f^{-1}_{(\Pi,\gamma)}$ are Lipschitz continuous at $\Pi$ under the distances $d^{MS}$ and $d^{KM}$. Consequently, we observe that convergent sequences of partition profiles in the MS topology are mapped by $f_{(\Pi,\gamma)}$ to convergent sequences of common priors in the KM topology. A converse statement holds for sequences of common priors that belong to $f_{(\Pi,\gamma)}(\mathcal{P}^N)$.

\begin{corollary}[Convergence under $f_{(\Pi,\gamma)}$]
If a sequence of partition profiles $(\Pi_n)$ converges to $\Pi$ in the MS topology, then the sequence of common priors $(f_{(\Pi,\gamma)}(\Pi_n))$ converges to $f_{(\Pi,\gamma)}(\Pi)$ in the KM topology. Moreover, if a sequence of common priors $(f_{(\Pi,\gamma)}(\Pi_n))$ converges to $f_{(\Pi,\gamma)}(\Pi)$ in the KM topology, then the sequence of partition profiles $(\Pi_n)$ converges to $\Pi$ in the MS topology.
\end{corollary}

The first part of Theorem \ref{thm_1} also immediately establishes one direction of \cite{kajii1998payoff}'s conjecture: for any $\epsilon \in (0,\gamma]$, if two partition profiles are in an $\epsilon$-neighborhood in the MS topology, then $f_{(\Pi,\gamma)}$ labels them so that the associated common priors are in an $\epsilon$-neighborhood in the KM topology. 
The second part of Theorem \ref{thm_1} formalizes and establishes the other direction: for any $\epsilon \in (0, \gamma]$, if it is possible to label partition profiles in such a way that the associated common priors are in an $\epsilon$-neighborhood in the KM topology, then these partition profiles must be in an $(N+1) \epsilon$-neighborhood in the MS topology. 

\section{Proof of Main Result}

The outline of the proof of Theorem \ref{thm_1} is as follows. We first construct a function $f_{(\Pi,\gamma)}: \mathcal{P}^N \rightarrow \Delta(S \times T)$ such that, for every $\Pi' \in \mathcal{P}^N$, $f_{(\Pi,\gamma)}(\Pi')$ is $\Pi'$-consistent (Section \ref{construct}). We then establish part 1(a) of Theorem \ref{thm_1} using this function (Section \ref{pf_part1a}). Finally, we jointly establish parts 1(b) and 2 of Theorem \ref{thm_1} (Section \ref{pf_part1b2}).

\subsection{Construction of $f_{(\Pi,\gamma)}$ and the Common Support Condition}\label{construct}

Fix a partition profile $\Pi \in \mathcal{P}^N$. We first define $f_{(\Pi,\gamma)}(\Pi) \in \Delta(S \times T)$. For each $i \in \mathcal{N}$, partition $T_i$ into countably infinite sets $X_i$ and $Y_i$. Let $\alpha_i: \Pi_i \rightarrow X_i$ be a one-to-one function. Such a function exists because $\Pi_i$ is countable and $X_i$ is countably infinite. Now, define $\tau: S \rightarrow T$ to be the $\Pi$-labeling
\[ \tau(s) := (\alpha_1(\Pi_1(s)), \ldots , \alpha_N(\Pi_N(s))). \] Then, set $f_{(\Pi,\gamma)}(\Pi)$ equal to the unique $\tau$-consistent common prior $\mu$.

We now define $f_{(\Pi,\gamma)}(\Pi')$ for each partition profile $\Pi' \neq \Pi$. For this purpose, let $\tilde{\Pi}_i := \{\text{$\pi \in \Pi_i$: $\pi \cap I_{\Pi,\Pi'}(\gamma) \neq \emptyset$}\}$. We claim that the function $\beta_i: \tilde{\Pi'}_i \rightarrow \tilde{\Pi}_i$ with $\beta_i(\Pi'_i(s))=\Pi_i(s)$ for all $s \in I_{\Pi,\Pi'}(\gamma)$ is well-defined and one-to-one. To prove $\beta_i$ is well-defined, suppose, towards contradiction, that there exist $s,s' \in I_{\Pi,\Pi'}(\gamma)$ such that $\Pi'_i(s)=\Pi'_i(s'):=\pi'$ and $\beta_i(\Pi'_i(s))=\Pi_i(s) \neq \Pi_i(s')=\beta_i(\Pi'_i(s'))$. Since $P(\pi'\backslash \Pi_i(s)  | \pi')+P(\Pi_i(s) \cap \pi' | \pi')=1$ and $s \in I_{\Pi,\Pi'}(\gamma)$ implies $P(\pi'\backslash \Pi_i(s)  | \pi') \leq \gamma < 1/2$, $P(\Pi_i(s) \cap \pi' | \pi')>1/2$ (note that the strict inequality is ensured by choosing $\gamma \in (0,1/2)$). Similarly, $P(\Pi_i(s') \cap \pi' | \pi')>1/2$. But because $\Pi_i(s) \cap \pi'$ and $\Pi_i(s') \cap \pi'$ are disjoint events, we must then have $P(\Pi_i(s) \cap \pi' | \pi')+P(\Pi_i(s') \cap \pi' | \pi')>1$, a contradiction. To prove $\beta_i$ is one-to-one, take two distinct partition elements  $\pi'_1,\pi'_2 \in \tilde{\Pi'}_i$. Suppose, towards contradiction, that $\beta_i(\pi'_1)=\beta_i(\pi'_2) := \pi \in \tilde{\Pi}_i$. Since $P(\pi \backslash \pi'_1| \pi)+ P(\pi \cap \pi'_1| \pi)=1$ and $\pi \in \tilde{\Pi}_i$ implies $P(\pi \backslash \pi'_1| \pi) \leq \gamma < 1/2$, $P(\pi \cap \pi'_1| \pi)>1/2$. Similarly, $P(\pi \cap \pi'_2| \pi)>1/2$. But because $\pi \cap \pi'_1$ and $\pi \cap \pi'_2$ are disjoint events, we must then have $P(\pi \cap \pi'_1|\pi)+P(\pi \cap \pi'_2|\pi)>1$, a contradiction. 

Now, let $B_i:=  \{ s \in S: \Pi'_i(s) \cap I_{\Pi,\Pi'}(\gamma) \neq \emptyset \}$ and $\alpha'_i: \Pi'_i \backslash \tilde{\Pi'}_i \rightarrow Y_i$ be a one-to-one function, where such a function $\alpha'_i$ once again exists because $\Pi'_i$ is countable and $Y_i$ is countably infinite. Define the $\Pi'$-labeling $\tau': S \rightarrow T$ by \[\tau'_i(s) := \begin{cases} \alpha_i ( \beta_i(\Pi'_i(s))) & \text{if $s \in B_i$} \\ \alpha'_i(\Pi'_i(s)) & \text{otherwise.} \end{cases}\] Then, set $f_{(\Pi,\gamma)}(\Pi')$ equal to the unique $\tau'$-consistent common prior $\mu'$.

From the preceding paragraphs, we have obtained a function $f_{(\Pi,\gamma)}$ such that, for every $\Pi' \in \mathcal{P}^N$, $f_{(\Pi,\gamma)}(\Pi')$ is $\Pi'$-consistent. This function also possesses a crucial property called the \textbf{common support condition (CSC)}: if $\mu=f_{(\Pi,\gamma)}(\Pi)$ is $\tau$-consistent with $\Pi$-labeling $\tau$ and $\mu'=f_{(\Pi,\gamma)}(\Pi')$ is $\tau'$-consistent with $\Pi'$-labeling $\tau'$, then $s \in I_{\Pi, \Pi'}(\gamma)$ implies $\tau(s)=\tau'(s)$. A simple example illustrates the condition.

\begin{figure}[!ht]
	\centering
	\subfloat[$f_{(\Pi,\gamma)}$ violates CSC.]{
		\begin{tikzpicture}[ every node/.style={align=center}]

% --- left column of row labels (aligned) ---
\def\labx{-3.4}
\node[anchor=east] (LabPi1)   at (\labx,2.75)    {$P$};
\node[anchor=east] (LabPi1)   at (\labx,2)    {$\Pi_1$};
\node[anchor=east] (LabTau1)  at (\labx,1.25) {$\tau_1$};
\node[anchor=east] (LabT1)    at (\labx,0.5)  {$T_1$};
\node[anchor=east] (LabTau1p) at (\labx,-0.25){$\tau'_1$};
\node[anchor=east] (LabPi1p)  at (\labx,-1.0) {$\Pi'_1$};

% --- top partition: one ellipse with two s_.. nodes ---
\node (pH) at (-2,2.75) {$1-\epsilon$};
\node (pL) at (0,2.75) {$\epsilon$};
\node (sH) at (-2,2) {$s_H$};
\node (sL) at ( 0,2) {$s_L$};

\draw (-1,2) ellipse (1.6cm and 0.45cm);

% --- middle: T1 nodes ---
\node (H)  at (-2,0.5) {$H$};
\node (Hp) at ( 0,0.5) {$L$};

% --- bottom partition: circles around s_.. (Pi1') ---
\node[draw,circle,minimum size=0.76cm] (sHp) at (-2,-1) {$s_H$};
\node[draw,circle,minimum size=0.76cm] (sLp) at ( 0,-1) {$s_L$};

% --- arrows from top ellipse down to T1 (direct) ---
\draw[->] ($(sH.south)+(1,-0.25)$) -- ($(H.north)+(0.2,0.02)$);

% --- arrows from bottom circles up to T1 (direct mapping) ---
\draw[->] ($(sHp.north)+(0,0.04)$) -- ( $(Hp.south)+(-0.2,0.1)$);
\draw[->] ($(sLp.north)+(0,0.04)$) -- ($(H.south)+(0.2,0.1)$);

\end{tikzpicture}\label{CSC_1}}
\quad \quad \quad \quad \subfloat[$f_{(\Pi,\gamma)}$ does not violate CSC.]{\begin{tikzpicture}[ every node/.style={align=center}]

% --- left column of row labels (aligned) ---
\def\labx{-3.4}
\node[anchor=east] (LabPi1)   at (\labx,2.75)    {$P$};
\node[anchor=east] (LabPi1)   at (\labx,2)    {$\Pi_1$};
\node[anchor=east] (LabTau1)  at (\labx,1.25) {$\tau_1$};
\node[anchor=east] (LabT1)    at (\labx,0.5)  {$T_1$};
\node[anchor=east] (LabTau1p) at (\labx,-0.25){$\tau'_1$};
\node[anchor=east] (LabPi1p)  at (\labx,-1.0) {$\Pi'_1$};

% --- top partition: one ellipse with two s_.. nodes ---
\node (pH) at (-2,2.75) {$1-\epsilon$};
\node (pL) at (0,2.75) {$\epsilon$};
\node (sH) at (-2,2) {$s_H$};
\node (sL) at ( 0,2) {$s_L$};

\draw (-1,2) ellipse (1.6cm and 0.45cm);

% --- middle: T1 nodes ---
\node (H)  at (-2,0.5) {$H$};
\node (Hp) at ( 0,0.5) {$L$};

% --- bottom partition: circles around s_.. (Pi1') ---
\node[draw,circle,minimum size=0.76cm] (sHp) at (-2,-1) {$s_H$};
\node[draw,circle,minimum size=0.76cm] (sLp) at ( 0,-1) {$s_L$};

% --- arrows from top ellipse down to T1 (direct) ---
\draw[->] ($(sH.south)+(1,-0.25)$) -- ($(H.north)+(0.2,0.02)$);

% --- arrows from bottom circles up to T1 (direct mapping) ---
\draw[->] ($(sHp.north)+(0,0.04)$) -- ($(H.south)+(0,-0.02)$);
\draw[->] ($(sLp.north)+(0,0.04)$) -- ($(Hp.south)+(0,-0.02)$);

\end{tikzpicture}\label{CSC_2}}
	\caption{Illustration of the common support condition (CSC).}\label{CSC_exmp}
\end{figure}
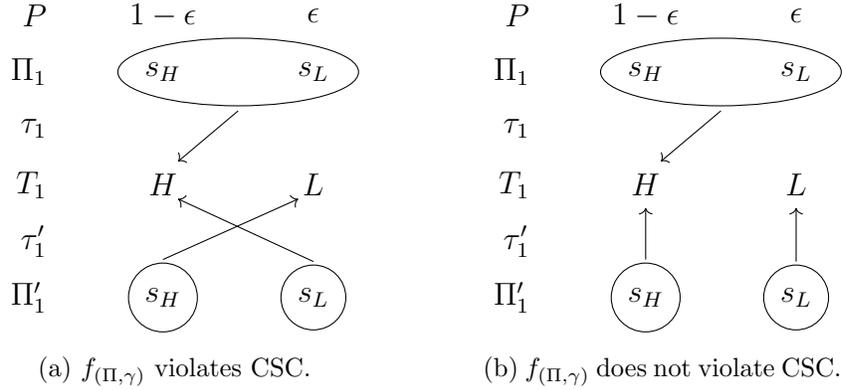

\begin{exmp}
In Figure \ref{CSC_exmp}, for simplicity of exposition, the state space is binary: $S:=\{s_H, s_L \}$. The prior distribution $P$ places  probability $\epsilon \in (0, \gamma]$ on $s_{L}$ and $1-\epsilon$ on $s_{H}$. The partitions for player 1 are $\Pi_1:=\{ \{s_{H}, s_{L} \}\}$ and $\Pi'_1:=\{ \{s_{H}\}, \{s_{L} \}\}$. Let $\mu=f_{(\Pi,\gamma)}(\Pi)$ be $\tau$-consistent with the $\Pi$-labeling $\tau$ satisfying $\tau_1(s_H)=\tau_1(s_L)=H$ and $\mu'=f_{(\Pi,\gamma)}(\Pi')$ be $\tau'$-consistent with some $\Pi'$-labeling $\tau'$. In Figure \ref{CSC_1}, $\tau'$ is chosen so that $f_{(\Pi,\gamma)}$ violates the common support condition; in particular, $\tau'_1(s_H)=L$ and $\tau'_1(s_L)=H$. To see why the condition is violated, notice that $$I_{\Pi,\Pi'}(\gamma)= \{s_{H} \}$$ because $$P(\Pi_1(s_{H}) \backslash \Pi'_1(s_{H})|\Pi_1(s_{H}))= P(s_{L}|\Pi_1(s_{H}))=\epsilon \leq \gamma$$ and $$P(\Pi_1(s_{L}) \backslash \Pi'_1(s_{L})|\Pi_1(s_{L}))= P(s_{H}|\Pi_1(s_{L}))=1-\epsilon >\gamma.$$ Nevertheless, the $\Pi$-labeling $\tau$ does not coincide with the $\Pi'$-labeling $\tau'$ at $s_H$: $\tau_1(s_{H})=H \neq L= \tau'_1(s_{H})$. On the other hand, the portion of the $\Pi'$-labeling $\tau'$ depicted in Figure \ref{CSC_2}, which sets $\tau'_1(s_H)=H$ and $\tau'_1(s_L)=L$, does not violate the common support condition because $\tau_1(s_{H})=H=\tau'_1(s_{H})$.
\end{exmp}

\subsection{Proof of Theorem \ref{thm_1}: part 1(a)}\label{pf_part1a}

We now show that if $\Pi$ and $\Pi'$ are close in the MS topology, then $f_{(\Pi,\gamma)}(\Pi)$ and $f_{(\Pi,\gamma)}(\Pi')$ are close in the KM topology. Formally, for $\epsilon \in (0,\gamma]$, $d^{MS}(\Pi, \Pi') \leq \epsilon$ implies $d^{KM}(f_{(\Pi,\gamma)}(\Pi) , f_{(\Pi,\gamma)}(\Pi')) \leq  \epsilon$.

A preliminary Lemma establishes the precise sense under which first-order beliefs of each type cohere with those in an associated partition model.

\begin{lemma}\label{prob_eq}
For any event $E \subseteq S$, $\Pi'$-labeling $\tau'$, and $\tau'$-consistent common prior $\mu'$, if $s \in S$, then
\begin{equation} P(E| \Pi'_i(s))=\mu'((\iota \times \tau')(E)| \tau'_i(s)). \notag \end{equation}

\end{lemma}

\begin{proof}
See Appendix \ref{proof_prob_eq}.
\end{proof}\medskip

\noindent Moreover, under the common support condition, the restrictions placed on conditional beliefs at partition elements in the MS topology bound the conditional beliefs of the types to which they are mapped.

\begin{lemma}\label{conditional}
Suppose $\mu =f_{(\Pi,\gamma)}(\Pi)$ is $\tau$-consistent with the $\Pi$-labeling $\tau$ and $\mu'=f_{(\Pi,\gamma)}(\Pi')$ is $\tau'$-consistent with the $\Pi'$-labeling $\tau'$.
For any $\epsilon \in (0,\gamma]$, if $s \in I_{\Pi,\Pi'}(\epsilon)$, then $(s,\tau'(s)) \in A_{\mu,\mu'}(\epsilon)$.
\end{lemma}

\begin{proof}
See Appendix \ref{proof_conditional}.
\end{proof} \medskip

\noindent The next Lemma concerns higher-order beliefs. Specifically, if the event $E \subseteq S$ is common $p$-belief under the partition profile $\Pi$, then $(\iota \times \tau')(E)$ is common $p$-belief under the common prior obtained from any $\Pi'$-consistent type function $\tau'$.

\begin{lemma}\label{consist}
For any event $E \subseteq S$, $\Pi'$-labeling $\tau'$, and $\tau'$-consistent common prior $\mu'$, if $s \in  C^p_{\Pi'}(E)$, then $(s, \tau'(s)) \in C^p_{\mu'}((\iota \times \tau')(E))$.
\end{lemma}

\begin{proof}
See Appendix \ref{proof_consist}. \end{proof}\medskip

Using the preceding lemmas, we finally show that if two partition profiles are close in the MS topology and are mapped to common priors consistent with a pair of type functions satisfying the common support condition, then the common priors are close in the KM topology.

\begin{lemma}\label{partitiontype}
Suppose $\mu =f_{(\Pi,\gamma)}(\Pi)$ is $\tau$-consistent with the $\Pi$-labeling $\tau$ and $\mu'=f_{(\Pi,\gamma)}(\Pi')$ is $\tau'$-consistent with the $\Pi'$-labeling $\tau'$.
For any $\epsilon \in (0,\gamma]$, if $d^{MS}(\Pi,\Pi') \leq \epsilon$, then $d^{KM}(\mu,\mu') \leq \epsilon$.
\end{lemma}

\begin{proof}
See Appendix \ref{proof_partitiontype}. \end{proof}\medskip

\noindent Theorem \ref{thm_1} part 1(a) follows directly from Lemma \ref{partitiontype}.

\subsection{Theorem \ref{thm_1}: parts 1(b) and 2}\label{pf_part1b2}

We now show that if $f_{(\Pi,\gamma)}(\Pi)$ and $f_{(\Pi,\gamma)}(\Pi')$ are close in the KM topology, then $\Pi$ and $\Pi'$ are close in the MS topology. Formally, for $\epsilon \in (0,\gamma]$, $d^{KM}(f_{(\Pi,\gamma)}(\Pi) , f_{(\Pi,\gamma)}(\Pi')) \leq  \epsilon$ implies $d^{MS}(\Pi, \Pi') \leq (N+1)\epsilon$. The proof mirrors the proof of part 1(a). 

A preliminary Lemma establishes the sense in which first-order beliefs at a partition element cohere with those at the type to which the partition element is mapped.

\begin{lemma}\label{prob_eq2}
For any event $E \subseteq S \times T$, $\Pi'$-labeling $\tau'$, and $\tau'$-consistent common prior $\mu'$, if $s \in S$, then
\begin{equation} \mu'(E | \tau'_i(s))=P((\iota \times \tau')^{-1}(E)| \Pi'_i(s)). \notag \end{equation}
\end{lemma}

\begin{proof}
See Appendix \ref{proof_prob_eq2}.
\end{proof}\medskip

\noindent Moreover, the restrictions placed on conditional beliefs on types in the KM topology bound the conditional beliefs at the partition elements to which they correspond.

\begin{lemma}\label{conditional2}
Suppose $\mu$ is $\tau$-consistent with the $\Pi$-labeling $\tau$ and $\mu'$ is $\tau'$-consistent with the $\Pi'$-labeling $\tau'$. For any $\epsilon \in (0,\gamma]$, if $ s \in (\iota \times \tau')^{-1}(A_{\mu,\mu'}(\epsilon))$ and $\tau(s)=\tau'(s)$, then $s \in I_{\Pi,\Pi'}(\epsilon)$ and $P(\{s' \in S: \tau_i(s') \neq \tau'_i(s')\}|\Pi'_i(s)) \leq \epsilon$ for all $i$.
\end{lemma}

\begin{proof}
See Appendix \ref{proof_conditional2}.
\end{proof}\medskip

\noindent The next Lemma concerns higher-order beliefs of events at which beliefs are similar.

\begin{lemma}\label{consist2}
Suppose $\mu$ is $\tau$-consistent with the $\Pi$-labeling $\tau$ and $\mu'$ is $\tau'$-consistent with the $\Pi'$-labeling $\tau'$. For any $\epsilon \in (0,\gamma]$, if $s \in (\iota \times \tau')^{-1}(C^{1-\epsilon}_{\mu'}(A_{\mu,\mu'}(\epsilon)) \cap A_{\mu,\mu'}(\epsilon))$  and $\tau(s)=\tau'(s)$, then $s \in C^{1-2\epsilon}_{\Pi'}(I_{\Pi,\Pi'}(\epsilon))$.
\end{lemma}

\begin{proof}
See Appendix \ref{proof_consist2}.
\end{proof} \medskip

 Using the preceding lemmas, we show that if two common priors are close in the KM topology, then the associated partition profiles are close in the MS topology.

\begin{lemma}\label{typespartitions}
Suppose $\mu$ is $\tau$-consistent with the $\Pi$-labeling $\tau$ and $\mu'$ is $\tau'$-consistent with the $\Pi'$-labeling $\tau'$.
For any $\epsilon \in (0,\gamma]$, if $d^{KM}(\mu,\mu') \leq \epsilon$, then $d^{MS}(\Pi,\Pi') \leq (N+1) \epsilon$.
\end{lemma}

\begin{proof}
See Appendix \ref{proof_typespartitions}.
\end{proof}\medskip

We now argue that Lemma \ref{typespartitions} immediately establishes Theorem \ref{thm_1} parts 1(b) and 2. Fix $\epsilon \in (0,\gamma]$. Observe that $f_{(\Pi,\gamma)}(\Pi')$ is $\Pi'$-consistent for all $\Pi' \in \mathcal{P}^N$. Hence, for any $\Pi, \Pi' \in \mathcal{P}^N$, there exists a $\Pi$-labeling $\tau$ and a $\Pi'$-labeling $\tau'$ such that $f_{(\Pi,\gamma)}(\Pi)$ is $\tau$-consistent and $f_{(\Pi,\gamma)}(\Pi')$ is $\tau'$-consistent. Then, from Lemma \ref{typespartitions}, $d^{KM}(f_{(\Pi,\gamma)}(\Pi), f_{(\Pi,\gamma)}(\Pi')) \leq \epsilon$ implies $d^{MS}(\Pi,\Pi') \leq (N+1) \epsilon$. Lemma \ref{typespartitions} also establishes Theorem \ref{thm_1} part 2 because its proof does not rely on any other properties of $f_{(\Pi,\gamma)}$, e.g., the common support condition. Specifically, if $\mu$ is $\Pi$-consistent and $\mu'$ is $\Pi'$-consistent, then there exists a $\Pi$-labeling $\tau$ and a $\Pi'$-labeling $\tau'$ such that $\mu$ is $\tau$-consistent and $\mu$ is $\tau'$-consistent.  Then, from Lemma \ref{typespartitions}, $d^{KM}(\mu,\mu') \leq \epsilon$ implies $d^{MS}(\Pi,\Pi') \leq (N+1) \epsilon$. 

\begin{appendix}
	
\section{Proofs of Lemmas}\label{proofs}

\subsection{Lemma \ref{prob_eq}}\label{proof_prob_eq}

Observe that $\Pi'_i(s)= (\tau')^{-1}( \{\tau'_i(s)\} \times T_{-i})$ because $\tau'$ is a $\Pi'$-labeling. Hence, for an arbitrary event $E \subseteq S$,
\[P(E \cap \Pi'_i(s))= \mu'(E \times \{\tau'_i(s)\} \times T_{-i})\] because $\mu'$ is $\tau'$-consistent and $\tau'$ is a $\Pi'$-labeling. Moreover,
\[\mu'(E \times \{\tau'_i(s)\} \times T_{-i})= \mu'(E \times \{\tau'_i(s)\} \times \tau'_{-i}(E))\]
because
\[ \mu'(E \times \{\tau'_i(s)\} \times (T_{-i} \backslash \tau'_{-i}(E)))= P(E \cap (\tau')^{-1}(\{\tau'_i(s)\} \times T_{-i} \backslash \tau'_{-i}(E)))=P(E \cap \emptyset)=0. \]
It follows that
\[ P(E| \Pi'_i(s)) = \frac{P(E \cap \Pi'_i(s))}{P(S \cap \Pi'_i(s))}=\frac{\mu'(E \times \{\tau'_i(s)\} \times \tau'_{-i}(E)) }{\mu'(S \times \{\tau'_i(s)\} \times T_{-i})}=\mu'((\iota \times \tau')(E)| \tau'_i(s)).\]

\subsection{Lemma \ref{conditional}}\label{proof_conditional}

If $s \in I_{\Pi,\Pi'}(\epsilon)$, then $\tau_i(s)=\tau'_i(s)=t_i$ for all $i$ by $\epsilon \in (0,\gamma]$ and the common support condition. Suppose, towards contradiction, that $|\mu(E \times F|t_i)-\mu'(E \times F| t_i)|> \epsilon$ for some event $E \times F \subseteq S \times T$ and some player $i$. Then, for $G= E \cap (\tau')^{-1}(F)$, it must be that
	\[|P(G|\Pi_i(s))-P(G|\Pi'_i(s))|> \epsilon\]
	because $\tau$ is a $\Pi$-labeling and $\tau'$ is a $\Pi'$-labeling. We show that this cannot occur and hence $(s,\tau'(s)) \in A_{\mu,\mu'}(\epsilon)$. 
 
 Towards contradiction, observe that 
\begin{equation}
		\begin{aligned}
			P(G| \Pi_i(s))- P(G| \Pi'_i(s)) =  
			  & \quad P( (\Pi_i(s) \backslash \Pi'_i(s)) \cap G| \Pi_i(s)) -P ((\Pi'_i(s) \backslash \Pi_i(s)) \cap G| \Pi'_i(s))+\\
			& \quad  P( (\Pi'_i(s) \cap \Pi_i(s)) \cap G| \Pi_i(s))-P((\Pi'_i(s) \cap \Pi_i(s)) \cap G| \Pi'_i(s))\\
   \leq & \quad  P( \Pi_i(s) \backslash \Pi'_i(s) \cap G| \Pi_i(s)) +\\
			& \quad P( (\Pi'_i(s) \cap \Pi_i(s)) \cap G| \Pi_i(s))-P((\Pi'_i(s) \cap \Pi_i(s)) \cap G| \Pi'_i(s)),
		\end{aligned}
		\notag
	\end{equation}
 where the inequality follows from $P (\Pi'_i(s) \backslash \Pi_i(s) \cap G | \Pi'_i(s)) \geq 0$. Moreover,
 \begin{equation}
     \begin{aligned}
P( (\Pi'_i(s) \cap \Pi_i(s)) \cap G| \Pi_i(s))-P((\Pi'_i(s) \cap \Pi_i(s)) \cap G | \Pi'_i(s))=\\
\frac{(P(\Pi'_i(s))-P(\Pi_i(s)))P(\Pi'_i(s) \cap \Pi_i(s) \cap G)}{P(\Pi_i(s))P(\Pi'_i(s))} \leq \\ \max\{\frac{(P(\Pi'_i(s))-P(\Pi_i(s)))P(\Pi'_i(s) \cap \Pi_i(s) )}{P(\Pi_i(s))P(\Pi'_i(s))},0\}=\\
\max\{P( \Pi'_i(s) \cap \Pi_i(s)| \Pi_i(s))-P(\Pi'_i(s) \cap \Pi_i(s) | \Pi'_i(s)),0\}.
     \end{aligned} \notag
 \end{equation} Hence,
\begin{equation}
		\begin{aligned}
			P(G| \Pi_i(s))- P(G| \Pi'_i(s)) \leq\\  
			  P( \Pi_i(s) \backslash \Pi'_i(s)| \Pi_i(s))+ \max\{P( \Pi'_i(s) \cap \Pi_i(s)| \Pi_i(s))-P(\Pi'_i(s) \cap \Pi_i(s)| \Pi'_i(s)),0\} = \\
    \max\{1- P(\Pi'_i(s) \cap \Pi_i(s)| \Pi'_i(s)), P( \Pi_i(s) \backslash \Pi'_i(s)| \Pi_i(s))\}.
		\end{aligned}
		\notag
	\end{equation}
From $s \in I_{\Pi,\Pi'}(\epsilon)$, we have  $P( \Pi_i(s) \backslash \Pi'_i(s)| \Pi_i(s)) \leq \epsilon$ and $P( \Pi'_i(s) \backslash \Pi_i(s)| \Pi'_i(s)) \leq \epsilon$.
Moreover, $P( \Pi'_i(s) \backslash \Pi_i(s)| \Pi'_i(s))+ P( \Pi'_i(s) \cap \Pi_i(s)| \Pi'_i(s))= 1$ implies $P(\Pi'_i(s) \cap \Pi_i(s) | \Pi'_i(s)) \geq 1-\epsilon$ by $P( \Pi'_i(s) \backslash \Pi_i(s)| \Pi'_i(s)) \leq \epsilon$. So,
\[P(G| \Pi_i(s))- P(G| \Pi'_i(s)) \leq \epsilon.\]
A symmetric argument ensures 
\[P(G| \Pi'_i(s))- P(G| \Pi_i(s)) \leq \epsilon.\]
Hence,
\[|P(G|\Pi_i(s))-P(G|\Pi'_i(s))| \leq \epsilon,\]
our desired contradiction.

\subsection{Lemma \ref{consist}}\label{proof_consist}

We show by induction that $s \in \cap_{m \geq 1} (B^p_{\Pi'})^m (E)= C^p_{\Pi'}(E)$ implies $(s,\tau'(s)) \in \cap_{m \geq 1} (B^p_{\mu'})^m ((\iota \times \tau')(E))= C^p_{\mu'}((\iota \times \tau')(E))$. For the base case, take $s \in B^p_{\Pi'}(E)$. Then, for all $i$, $$ \mu'((\iota \times \tau')(E)| \tau'_i(s))=P(E| \Pi'_i(s)) \geq p,$$ where the first equality is from Lemma \ref{prob_eq} and the inequality is from the definition of $B^p_{\Pi'}(E)$. Hence, $(s,\tau'(s)) \in (B^{p}_{\mu'}) ((\iota \times \tau')(E))$. 

The induction hypothesis is that if $s \in (B^p_{\Pi'})^m(E)$, then $(s,\tau'(s)) \in (B^{p}_{\mu'})^m ((\iota \times \tau')(E))$. Take $s \in (B^p_{\Pi'})^{m+1}(E)$. Then, for all $i$, $$ \mu'((\iota \times \tau')((B^{p}_{\Pi'})^{m} (E)))=P((B^{p}_{\Pi'})^{m} (E)|\Pi'_i(s)) \geq p,$$ where the equality again follows from Lemma \ref{prob_eq}. By the induction hypothesis, if $s \in (B^{p}_{\Pi'})^{m} (E)$, then $(s,\tau'(s)) \in (B^{p}_{\mu'})^m ((\iota \times \tau')(E))$. It follows that
\[(\iota \times \tau')((B^{p}_{\Pi'})^{m} (E)) \subseteq (B^{p}_{\mu'})^m ((\iota \times \tau')(E))  \] and, therefore, for all $i$,
$$\mu'((B^{p}_{\mu'})^m ((\iota \times \tau')(E))| \tau'_i(s)) \geq \mu'((\iota \times \tau')((B^{p}_{\Pi'})^{m} (E))| \tau'_i(s)) \geq p.$$ Hence, $(s,\tau'(s)) \in  (B^{p}_{\mu'})^{m+1} ((\iota \times \tau')(E))$. For every $m$ and every $s \in (B^p_{\Pi'})^m(E)$, we have $(s,\tau'(s)) \in (B^p_{\mu'})^m((\iota \times \tau')(E))$; hence, $s \in  C^p_{\Pi'}(E)$ implies $(s, \tau'(s)) \in C^p_{\mu'}((\iota \times \tau')(E))$.

\subsection{Lemma \ref{partitiontype}}\label{proof_partitiontype}

	If $d^{MS}(\Pi,\Pi') \leq \epsilon$, then $d^{MS}_1(\Pi, \Pi') \leq \epsilon$ and \[P(C^{1-\epsilon}_{\Pi'}(I_{\Pi,\Pi'}(\epsilon)) \cap I_{\Pi,\Pi'}(\epsilon)) \geq 1-\epsilon.\]
  Because $\mu'$ is $\tau'$-consistent, 
  \[\mu'(C^{1-\epsilon}_{\mu'}(A_{\mu,\mu'}(\epsilon)) \cap A_{\mu,\mu'}(\epsilon)) = P( (\iota \times \tau')^{-1}( C^{1-\epsilon}_{\mu'}(A_{\mu,\mu'}(\epsilon)) \cap A_{\mu,\mu'}(\epsilon))).\]
  Hence, it suffices to show
 \begin{equation}   (\iota \times \tau')(C^{1-\epsilon}_{\Pi'}(I_{\Pi,\Pi'}(\epsilon)) \cap I_{\Pi,\Pi'}(\epsilon)) \subseteq C^{1-\epsilon}_{\mu'}(A_{\mu,\mu'}(\epsilon)) \cap A_{\mu,\mu'}(\epsilon)
 \label{inclusion1} \end{equation}
 to establish $\rho^{KM}_1(\mu,\mu') \leq  \epsilon$, i.e.,
 \[\mu'(C^{1-\epsilon}_{\mu'}(A_{\mu,\mu'}(\epsilon)) \cap A_{\mu,\mu'}(\epsilon)) \geq 1-\epsilon. \] To prove \eqref{inclusion1}, observe from Lemma \ref{conditional} that if $s \in I_{\Pi,\Pi'}(\epsilon)$, then $(s, \tau'(s)) \in A_{\mu,\mu'}(\epsilon)$. We claim, also, that if $s \in C^{1-\epsilon}_{\Pi'}(I_{\Pi,\Pi'}(\epsilon))$, then $(s,\tau'(s)) \in C^{1-\epsilon}_{\mu'}(A_{\mu,\mu'}(\epsilon))$. Take $s \in C^{1-\epsilon}_{\Pi'}(I_{\Pi,\Pi'}(\epsilon))$. Then, by Lemma \ref{consist}, $(s,\tau'(s)) \in C^{1-\epsilon}_{\mu'}((\iota \times \tau')(I_{\Pi,\Pi'}(\epsilon)))$. By Lemma \ref{conditional}, $(\iota \times \tau')(I_{\Pi,\Pi'}(\epsilon)) \subseteq A_{\mu, \mu'}(\epsilon)$ and hence $(s,\tau'(s)) \in C^{1-\epsilon}_{\mu'}(A_{\mu,\mu'}(\epsilon))$ because $C^{1-\epsilon}_{\mu'}(E) \subseteq C^{1-\epsilon}_{\mu'}(F)$ for any events $E \subseteq F$. We have thus established \eqref{inclusion1} 
 and, therefore, $\rho^{KM}_1(\mu,\mu') \leq \epsilon$. A symmetric argument establishes $\rho^{KM}_1(\mu',\mu) \leq \epsilon$.

 It remains to show that $\rho^{KM}_0(\mu,\mu') \leq \epsilon$. Let $E \times F \subseteq S \times T$ be an arbitrary rectangle. Then,
\[|\mu(E \times F)- \mu'(E \times F)|= |P(E \cap \tau^{-1}(F))-P(E \cap (\tau')^{-1}(F))|\]
because $\mu$ is $\tau$-consistent and  $\mu'$ is $\tau'$-consistent. By the triangle inequality, we have
\begin{equation}
	\begin{aligned}
	|P(E \cap \tau^{-1}(F))-P(E \cap (\tau')^{-1}(F))| &\leq  \underbrace{|P(E \cap I_{\Pi,\Pi'}(\epsilon) \cap \tau^{-1}(F))-P(E \cap I_{\Pi,\Pi'}(\epsilon) \cap (\tau')^{-1}(F))|}_{(i)} \\
	& \quad + \underbrace{| P(E \backslash I_{\Pi,\Pi'}(\epsilon) \cap \tau^{-1}(F))-P(E \backslash  I_{\Pi,\Pi'}(\epsilon) \cap (\tau')^{-1}(F))|}_{(ii)}.
	\end{aligned}
	\notag 
\end{equation}
Term (i) equals zero by the common support condition: for every $s \in I_{\Pi,\Pi'}(\epsilon)$, we have $\tau(s)=\tau'(s)$. Hence,
\[\underbrace{\{s \in E \cap I_{\Pi,\Pi'}(\epsilon): \tau(s) \in F\}}_{= E \cap I_{\Pi,\Pi'}(\epsilon) \cap \tau^{-1}(F)}= \underbrace{ \{s \in E \cap I_{\Pi,\Pi'}(\epsilon): \tau'(s) \in F \}}_{ =E \cap I_{\Pi,\Pi'}(\epsilon) \cap (\tau')^{-1}(F)}.\] If $d^{MS}(\Pi,\Pi') \leq \epsilon$, then term $(ii)$ is less than $\epsilon$ because  \[\max\{P(E \backslash I_{\Pi,\Pi'}(\epsilon) \cap \tau^{-1}(F)), P(E \backslash  I_{\Pi,\Pi'}(\epsilon) \cap (\tau')^{-1}(F)) \} \leq P(S \backslash I_{\Pi,\Pi'}(\epsilon))\] and \[1-P(S \backslash I_{\Pi,\Pi'}(\epsilon))= P(I_{\Pi,\Pi'}(\epsilon)) \geq   P( C^{1-\epsilon}_{\Pi'}(I_{\Pi,\Pi'}(\epsilon)) \cap I_{\Pi,\Pi'}(\epsilon)) \geq 1-\epsilon.\] So,
\[P(E \backslash I_{\Pi,\Pi'}(\epsilon) \cap \tau^{-1}(F)) \leq \epsilon \]
and
\[P(E \backslash  I_{\Pi,\Pi'}(\epsilon) \cap (\tau')^{-1}(F)) \leq \epsilon.\]
Hence,
\[|P(E \backslash I_{\Pi,\Pi'}(\epsilon) \cap \tau^{-1}(F))-P(E \backslash  I_{\Pi,\Pi'}(\epsilon) \cap (\tau')^{-1}(F))| \leq \epsilon.\] We have thus shown
\[|\mu(E \times F)- \mu'(E \times F)|= |P(E \cap \tau^{-1}(F))-P(E \cap (\tau')^{-1}(F))| \leq \epsilon.\]
Moreover, because $E \times F \subseteq S \times T$ was an arbitrary rectangle, we have established \[\rho^{KM}_0(\mu,\mu')=\underset{E \times F \subseteq S \times T}{\sup}~|\mu(E \times F)-\mu'(E \times F)| \leq \epsilon.\] It follows that $d^{KM}(\mu,\mu') \leq \epsilon$.

\subsection{Lemma \ref{prob_eq2}}\label{proof_prob_eq2}

Because $\mu'$ is $\tau'$-consistent and $\tau'$ is a $\Pi'$-labeling,
\[\mu'( E  \cap (S \times \{\tau'_i(s)\} \times T_{-i}) )=  P((\iota \times \tau')^{-1}(E) \cap \Pi'_i(s))\]
and
\[\mu'(S \times \{\tau'_i(s)\} \times T_{-i})= P(\Pi'_i(s)).\]
It follows from the definition of conditional probability that
\begin{equation}
    \begin{aligned}
\mu'(E | \tau'_i(s)) &=
\frac{\mu'( E \cap (S \times \{\tau'_i(s)\} \times T_{-i}) )}{\mu'(S \times \{\tau'_i(s)\} \times T_{-i})} \\ &=\frac{P((\iota \times \tau')^{-1}(E) \cap \Pi'_i(s))}{P(\Pi'_i(s))}\\ &=P((\iota \times \tau')^{-1}(E)| \Pi'_i(s)). \notag
    \end{aligned}
\end{equation}

\subsection{Lemma \ref{conditional2}}\label{proof_conditional2}

We prove the contrapositive of the first statement in the Lemma. Suppose $s \not \in I_{\Pi,\Pi'}(\epsilon)$. Then, without loss of generality, there exists a player $i$ under which 
\[P(\Pi_i(s) \backslash \Pi'_i(s)| \Pi_i(s)) > \epsilon.\]
But,
\[ \mu( \Pi_i(s) \backslash \Pi'_i(s) \times T | \tau_i(s) )= P(\Pi_i(s) \backslash \Pi'_i(s)| \Pi_i(s))  \]
and 
\[ \mu'( \Pi_i(s) \backslash \Pi'_i(s) \times T| \tau'_i(s) )=0.\] So, if $\tau'_i(s)=\tau_i(s)=t_i$, then
\[ |\mu( \Pi_i(s) \backslash \Pi'_i(s) \times T | t_i)- \mu'( \Pi_i(s) \backslash \Pi'_i(s) \times T| t_i)|> \epsilon. \]
Hence, $(s,\tau'(s)) \not \in A_{\mu,\mu'}(\epsilon)$.

We now show that if $ s \in (\iota \times \tau')^{-1}(A_{\mu,\mu'}(\epsilon))$ and $\tau(s)=\tau'(s)$, then $P(D_i|\Pi'_i(s)) \leq \epsilon$, where $D_i:=\{s' \in S: \tau_i(s') \neq \tau'_i(s')\}$. Observe that  $s \in (\iota \times \tau')^{-1}(A_{\mu,\mu'}(\epsilon))$ and $\tau(s)=\tau'(s)=t$ ensures, for all $i$,
\[|\mu(E_i \times T | t_i)- \mu'(E_i \times T|t_i)| \leq \epsilon,\] where $E_i:=\{s' \in D_i: \tau'_i(s')=t_i \}$. But, using Lemma \ref{prob_eq2} and $\tau'$-consistency of $\mu'$,
\[\mu'(E_i \times T|t_i)=P(D_i| \Pi'_i(s))\]
and
\[\mu(E_i \times T | t_i)=0.\]
Thus, $P(D_i|\Pi'_i(s)) \leq \epsilon$.

\subsection{Lemma \ref{consist2}}\label{proof_consist2}
Suppose $s \in (\iota \times \tau')^{-1}(\cap_{m \geq 1} (B^{1-\epsilon}_{\mu'})^m (A_{\mu,\mu'}(\epsilon)) \cap A_{\mu,\mu'}(\epsilon))=(\iota \times \tau')^{-1}(C^{1-\epsilon}_{\mu'}(A_{\mu,\mu'}(\epsilon)) \cap A_{\mu,\mu'}(\epsilon))$ and $\tau(s)=\tau'(s)=t$. We show by induction that $s \in \cap_{m \geq 1} (B^{1-2\epsilon}_{\Pi'})^m (I_{\Pi,\Pi'}(\epsilon))=C^{1-2\epsilon}_{\Pi'}(I_{\Pi,\Pi'}(\epsilon))$. 

For the base case, take any $s \in S$ such that $(s,\tau'(s)) \in B^{1-\epsilon}_{\mu'}(A_{\mu,\mu'}(\epsilon)) \cap A_{\mu,\mu'}(\epsilon)$ and $\tau(s)=\tau'(s)=t$. Then, for all $i$, $\mu'(A_{\mu,\mu'}(\epsilon)|t_i) \geq 1-\epsilon$. By $\tau'$-consistency of $\mu'$ and Lemma \ref{prob_eq2}, for all $i$,
\[P((\iota \times \tau')^{-1}(A_{\mu,\mu'}(\epsilon))| \Pi'_i(s))=\mu'(A_{\mu,\mu'}(\epsilon)|t_i) \geq 1-\epsilon. \] By Lemma \ref{conditional2}, if $s \in (\iota \times \tau')^{-1}(A_{\mu,\mu'}(\epsilon)) $ and $\tau(s)=\tau'(s)$, then $s \in I_{\Pi,\Pi'}(\epsilon)$ and $P(\{s' \in S: \tau_i(s') \neq \tau'_i(s')\}|\Pi'_i(s)) \leq \epsilon$. So,
\[P(I_{\Pi,\Pi'}(\epsilon)|\Pi'_i(s)) \geq P((\iota \times \tau')^{-1}(A_{\mu,\mu'}(\epsilon))| \Pi'_i(s))- P(\{s' \in S: \tau_i(s') \neq \tau'_i(s')\}|\Pi'_i(s)) \geq 1-2\epsilon.  \]
That is, $s \in B^{1-2\epsilon}_{\Pi'}(I_{\Pi,\Pi'}(\epsilon))$.

The induction hypothesis is that for any $s \in S$ such that $(s,\tau'(s)) \in (B^{1-\epsilon}_{\mu'})^m (A_{\mu,\mu'}(\epsilon)) \cap A_{\mu,\mu'}(\epsilon)$ and $\tau(s)=\tau'(s)=t$, we have that $s \in (B^{1-2\epsilon}_{\Pi'})^m(I_{\Pi,\Pi'}(\epsilon))$. Suppose $(s,\tau'(s)) \in (B^{1-\epsilon}_{\mu'})^{m+1} (A_{\mu,\mu'}(\epsilon)) \cap A_{\mu,\mu'}(\epsilon)$ and $\tau(s)=\tau'(s)=t$. Then, for all $i$, $\mu'((B^{1-\epsilon}_{\mu'})^m (A_{\mu,\mu'}(\epsilon))|t_i) \geq 1-\epsilon$. By $\tau'$-consistency of $\mu'$ and Lemma \ref{prob_eq2}, for all $i$,
\[P((\iota \times \tau')^{-1}((B^{1-\epsilon}_{\mu'})^m (A_{\mu,\mu'}(\epsilon)))| \Pi'_i(s))=\mu'((B^{1-\epsilon}_{\mu'})^m (A_{\mu,\mu'}(\epsilon))|t_i) \geq 1-\epsilon. \]
By the induction hypothesis and $P(\{s' \in S: \tau_i(s') \neq \tau'_i(s')\}|\Pi'_i(s)) \leq \epsilon$,
\begin{equation}
    \begin{aligned}
        P((B^{1-2\epsilon}_{\Pi'})^m(I_{\Pi,\Pi'}(\epsilon))|\Pi'_i(s))\\ \geq P((\iota \times \tau')^{-1}((B^{1-\epsilon}_{\mu'})^m (A_{\mu,\mu'}(\epsilon)))| \Pi'_i(s))- P(\{s' \in S: \tau_i(s') \neq \tau'_i(s')\}|\Pi'_i(s))\\ \geq 1-2\epsilon.  
    \end{aligned} \notag
\end{equation}
That is, $s \in (B^{1-2\epsilon}_{\Pi'})^{m+1}(I_{\Pi,\Pi'}(\epsilon))$.

\subsection{Lemma \ref{typespartitions}}\label{proof_typespartitions}

To ease notation, let
\[A:= (\iota \times \tau')^{-1}(C^{1-\epsilon}_{\mu'}(A_{\mu,\mu'}(\epsilon)) \cap A_{\mu,\mu'}(\epsilon)),\]
\[B:=C^{1-\epsilon}_{\mu'}(A_{\mu,\mu'}(\epsilon)) \cap A_{\mu,\mu'}(\epsilon), \]
\[\underset{T_i}{\marg}~B :=\{t_i \in T_i: \exists (s,t_{-i}) ~\text{such that}~(s,t_i, t_{-i}) \in B\},\]
\[D:= \{s \in S: \tau(s) \neq \tau'(s)\},\]
and
\[D_i:= \{s \in S: \tau_i(s) \neq \tau'_i(s)\}.\]

We begin by proving two preliminary claims. 

\begin{claim}\label{lem_claim}
    \[ A \backslash D \subseteq C^{1-2\epsilon}_{\Pi'}(I_{\Pi,\Pi'}(\epsilon)) \cap I_{\Pi,\Pi'}(\epsilon) .\] 
\end{claim} 

\begin{proof}
By Lemma \ref{conditional2}, if $s \in (\iota \times \tau')^{-1}(C^{1-\epsilon}_{\mu'}(A_{\mu,\mu'}(\epsilon)) \cap A_{\mu,\mu'}(\epsilon))$ and $\tau(s)=\tau'(s)$, then $s \in I_{\Pi,\Pi'}(\epsilon)$. Moreover, by Lemma \ref{consist2}, if $s \in (\iota \times \tau')^{-1}(C^{1-\epsilon}_{\mu'}(A_{\mu,\mu'}(\epsilon)) \cap A_{\mu,\mu'}(\epsilon))$ and $\tau(s)=\tau'(s)$, then $s \in C^{1-2\epsilon}_{\Pi'}(I_{\Pi,\Pi'}(\epsilon))$. Hence, $s \in A \backslash D$ implies $s \in C^{1-2\epsilon}_{\Pi'}(I_{\Pi,\Pi'}(\epsilon)) \cap I_{\Pi,\Pi'}(\epsilon) $. \end{proof}

\begin{claim}\label{lem_claim2}
\[P(D \cap A) \leq N \epsilon.\]
\end{claim}
\begin{proof}
 Observe that for any $i$ and $t_i  \in \underset{T_i}{\marg}~B$,
\[ | \mu( (D_i \cap A \cap \{s \in S: \tau'_i(s)=t_i\}) \times T|t_i)-\mu'( (D_i \cap A \cap \{s \in S: \tau'_i(s)=t_i\}) \times T|t_i)| \leq \epsilon \]
by the definition of $A_{\mu,\mu'}(\epsilon)$. Observe that
\[ \mu( (D_i \cap A \cap \{s \in S: \tau'_i(s)=t_i\}) \times T|t_i)=0 \]
because $D_i \cap \{s \in S: \tau'_i(s)=t_i\} \subseteq \{s \in S: \tau_i(s) \neq t_i \}.$ Moreover,
\[\mu'((D_i \cap A) \times T| t_i)=\mu'( (D_i \cap A \cap \{s \in S: \tau'_i(s)=t_i\}) \times T|t_i).\]
So, for any $i$ and $t_i  \in \underset{T_i}{\marg}~B$,
\[\mu'( (D_i \cap A) \times T| t_i) \leq \epsilon.\] Note, also, for any $i$ and $t_i \not \in \underset{T_i}{\marg}~B$, 
\[\mu'((D_i \cap A) \times T| t_i)=0\]
because $\{s \in A=(\iota \times \tau')^{-1}(B): \tau'_i(s)=t_i\}=\emptyset$. Hence, by $\tau'$-consistency of $\mu'$, for any $i$,
\begin{equation}
\begin{aligned}
    P(D_i \cap A) &=\mu'( (D_i \cap A) \times T) \\
    &= \sum_{ t_i  \in \underset{T_i}{\marg}~B} \mu'(t_i) \underbrace{\mu'((D_i \cap A) \times T| t_i)}_{\leq \epsilon}+   \sum_{ t_i \not \in \underset{T_i}{\marg}~B} \mu'(t_i) \underbrace{\mu'((D_i \cap A) \times T| t_i)}_\text{$=0$}\\
    &\leq \epsilon.
\end{aligned} \notag
\end{equation} Therefore,
\[P(D \cap A) \leq \sum^N_{i=1} P(D_i \cap A) \leq N \epsilon.\]
\end{proof}

Now, observe that if $d^{KM}(\mu, \mu')\leq \epsilon$, then $\mu'(B) \geq 1-\epsilon$. Because $\mu'$ is $\tau'$-consistent,
\[P(A)=\mu'(B) \geq 1-\epsilon. \] Because $A \backslash D \subseteq C^{1-2\epsilon}_{\Pi'}(I_{\Pi,\Pi'}(\epsilon)) \cap I_{\Pi,\Pi'}(\epsilon)$ by Claim \ref{lem_claim}, we have 
\[P(C^{1-2\epsilon}_{\Pi'}(I_{\Pi,\Pi'}(\epsilon)) \cap I_{\Pi,\Pi'}(\epsilon)) \geq P(A\backslash D)=P(A)-P(D \cap A).\]
Because $P(D \cap A) \leq N \epsilon$ by Claim \ref{lem_claim2} and $P(A) \geq 1-\epsilon$, we thus have that
\[P(C^{1-2\epsilon}_{\Pi'}(I_{\Pi,\Pi'}(\epsilon)) \cap I_{\Pi,\Pi'}(\epsilon)) \geq  1-\epsilon -N\epsilon = 1-(N+1)\epsilon. \] Because \[C^{1-(N+1) \epsilon}_{\Pi'}(I_{\Pi,\Pi'}((N+1) \epsilon)) \cap I_{\Pi,\Pi'}((N+1) \epsilon) \supseteq C^{1-2\epsilon}_{\Pi'}(I_{\Pi,\Pi'}(\epsilon)) \cap I_{\Pi,\Pi'}(\epsilon),  \]
it follows that
\[P(C^{1-(N+1) \epsilon}_{\Pi'}(I_{\Pi,\Pi'}((N+1) \epsilon)) \cap I_{\Pi,\Pi'}((N+1) \epsilon)) \geq P(C^{1-2\epsilon}_{\Pi'}(I_{\Pi,\Pi'}(\epsilon)) \cap I_{\Pi,\Pi'}(\epsilon)) \geq 1-(N+1)\epsilon,\] i.e., $d^{MS}_1(\Pi,\Pi') \leq (N+1) \epsilon$. A symmetric argument ensures $d^{MS}_1(\Pi',\Pi) \leq (N+1)\epsilon$. Thus, $d^{MS}(\Pi,\Pi') \leq (N+1) \epsilon$.

\end{appendix}

\bibliography{bib}

\end{document}